\documentclass[11pt]{article}
\title{The multiplayer Colonel Blotto game}
\author{Enric Boix-Adser\`a\thanks{Email: \texttt{eboix@mit.edu}. Supported in part by an NSF GRFP fellowship and a Siebel Scholarship.}\\ MIT
	\and Benjamin L. Edelman\thanks{Email: \texttt{bedelman@g.harvard.edu}. Supported in part by NSF Grant CCF-15-09178.} \\ Harvard University
	\and Siddhartha Jayanti\thanks{Email: \texttt{jayanti@mit.edu}. Supported by an NDSEG Fellowship from the United States Department of Defense.\vspace{0.5em}} \\ \includegraphics[height=1.5em]{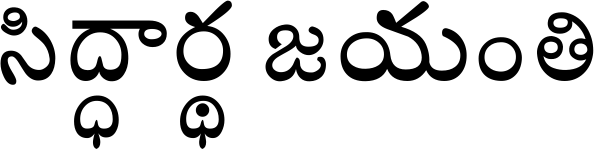} \\ MIT CSAIL}
\date{February 2020}

\usepackage{amsmath}
\usepackage{amsthm}
\usepackage{amssymb}
\usepackage{bbm}
\usepackage[table]{xcolor}
\usepackage{mathtools}
\usepackage{makecell}
\usepackage{bm}
\usepackage{todonotes}
\usepackage{diagbox}
\usepackage{tablefootnote}
\usepackage[ruled,vlined,linesnumbered]{algorithm2e}

\usepackage{tabularx}
\definecolor{lightgray}{gray}{0.9}
\newcolumntype{Y}{>{\centering\arraybackslash}X}

\usepackage[utf8]{inputenc}
\usepackage[english]{babel}
\usepackage{csquotes} 
\usepackage[margin=1in]{geometry}
\usepackage{xfrac}
\usepackage{enumitem}

\usepackage{microtype}

\usepackage[hidelinks]{hyperref}
\usepackage{thm-restate}

\newcommand{\E}[0]{\mathbb{E}} 
\newcommand{\R}[0]{\mathbb{R}} 
\newcommand{\one}[0]{\mathbbm{1}}

\newcommand{\B}[0]{\mathcal{B}} 
\newcommand{\cD}[0]{\mathcal{D}}

\DeclareMathOperator{\poly}{poly}

\DeclareMathOperator{\trace}{trace}

\DeclareMathOperator{\Ber}{Ber}

\let\baraccent=\= 
\renewcommand{\=}[1]{\stackrel{#1}{=}} 

\providecommand{\RR}{\mathbb{R}}
\providecommand{\NN}{\mathbb{N}}

\providecommand{\ZZ}{\mathbb{Z}}

\providecommand{\cB}{\mathcal{B}}
\providecommand{\B}{\mathcal{B}}

\providecommand{\PP}{\mathbb{P}}
\providecommand{\EE}{\mathbb{E}}

\providecommand{\eps}{\epsilon}

\providecommand{\N}{\mathbb{N}}

\providecommand{\NP}{\mathsf{NP}}

\mathchardef\mhyphen="2D 

\newcommand{\lone}[1]{\left\| {#1} \right\|_{1}}

\providecommand{\sm}{\setminus}

\newcommand{\interior}[1]{%
  {\kern0pt#1}^{\mathrm{o}}%
}

\newtheorem{theorem}{Theorem}[section]
\newtheorem{lemma}[theorem]{Lemma}
\newtheorem{definition}[theorem]{Definition}

\newtheorem{remark}[theorem]{Remark}

\newtheorem{claim}[theorem]{Claim}
\newtheorem{proposition}[theorem]{Proposition}
\newtheorem{corollary}[theorem]{Corollary}

\let\oldnl\nl
\newcommand{\nonl}{\renewcommand{\nl}{\let\nl\oldnl}}

\title{The Multiplayer Colonel Blotto Game}

\begin{document}

\maketitle

\begin{abstract}
 We initiate the study of the natural multiplayer generalization of the classic continuous \emph{Colonel Blotto game}. 
 The two-player Blotto game, introduced by Borel \cite{borel} as a model of resource competition across $n$ simultaneous fronts, has been studied extensively for a century and has seen numerous applications throughout the social sciences. 
 Our work defines the \emph{multiplayer Colonel Blotto game} and derives Nash equilibria for various settings of $k$ (number of players) and $n$. 
 We also introduce a ``Boolean'' version of Blotto that becomes interesting in the multiplayer setting. The main technical difficulty of our work, as in the two-player theoretical literature, is the challenge of coupling various marginal distributions into a joint distribution satisfying a strict sum constraint. In contrast to previous works in the continuous setting, we derive our couplings algorithmically in the form of efficient sampling algorithms.


 
\end{abstract}

\section{Introduction}

The \emph{Colonel Blotto game} has been featured in the game theory literature ever since it was introduced by Borel in 1921 \cite{borel}. It has found numerous applications in the social sciences as a model of competition with limited resources across simultaneous winner-take-all fronts. 

The basic structure of the game is as follows. 
There are two players, Alice and Bob, competing over $n$ \emph{battlefields} of value $v_1,\ldots,v_n$ (which may represent items, voting districts, advertising slots, etc.). 
Alice and Bob each have finite {\em budgets}---$\B_{Alice}, \B_{Bob}$---of a resource to distribute across the battlefields. 
They must simultaneously decide how to allot their budgets of the resource across the battlefields by placing a vector of $n$ bids, one for each battlefield. 
The value of each battlefield is won by the player that allocates more resources to it, or split evenly in the case of a tie. 
The players have the goal of maximizing the total value of their winnings.
It is common to restrict the game to be {\em symmetric}---players have the same budget---and/or {\em homogeneous}---all battlefields have the same value.

Though the game is simple to describe, there is considerable complexity in the equilibrium strategies that emerge.\footnote{It is well known that even the simplest Blotto games do not admit pure Nash equilibria.
Consider the two-player symmetric homogeneous Blotto game with $n > 2$ battlefields.
If Alice fixes any bid vector $\vec{a} = (a_1,\ldots,a_n)$ (where $a_1 \ne 0$ without loss of generality), then Bob can maximize his winnings by picking the action $\vec{b} = (0,a_2 + \eps, a_3 + \eps,\ldots,a_n + \eps)$, where $\eps = \frac{a_1}{n-1}$, to win all but the first battlefield.
This pair of actions is not in equilibrium because Alice can switch her strategy to that of Bob in order to win half of the total value rather than $1/n$ of it.
Therefore, in general we are looking for mixed Nash equilibria.}
Analysis of two-player Blotto has proved to be a challenging mathematical task, because randomized strategies for the game are complicated joint distributions over $n$-dimensional vectors on a simplex. 
However, there has been substantial recent progress in finding and classifying equilibria for several standard versions of the game \cite{laslier,weinstein,roberson,hortala,schwartz,thomas}, many of which are now essentially solved \cite{kovenock}.
In most cases, equilibria for the Blotto game have been developed based on solutions to the much simpler soft-budget constraint version of the game called {\em General Lotto}.
In a strategy for the Lotto game, each player bids a {\em distribution} for each battlefield, rather than a single value, and the winner of the battlefield is computed by comparing single samples from the distributions played by the two players.
What makes Lotto easier to analyze is its budget constraint, that the sum of the $n$ sampled bids of each player $i$ is at most $\B_i$ in {\em expectation}.
In contrast, the Blotto game requires a way to couple the $n$ different bid distributions such that any joint sample satisfies the budget constraint $\B_i$ {\em with probability 1}.

Modeling two-party elections is a famous application of the Blotto game \cite{laslier, roberson, merolla}.
Hoping to understand multiparty electoral systems, Myerson alluded to a Blotto game with more than two players in \cite{myerson}, which compares different types of multiparty election systems by studying the equilibrium strategies those systems induce.
In this context, the classic plurality vote elections conducted in many parliamentary democracies such as India and the United Kingdom are naturally modelled by a multiplayer generalization of Colonel Blotto with, e.g., the multiple parties corresponding to players, voting districts corresponding to battlefields, and district advertising expenditures corresponding to the resource allocations.
However, stating that ``the hardest part of [the Blotto] problem was to construct joint distributions for allocations that always sum to the given total,'' Myerson weakened the true budget constraint to the soft one and only analyzed what would nowadays be called multiplayer homogeneous symmetric General Lotto.
While Lotto is a good approximation to Blotto in the regime of large $n$ (by law of large numbers), 
it is a rather poor approximation in the regime of small $n$.
Nevertheless, analyzing multiplayer Blotto has remained an open problem for nearly 30 years.

\subsection{Our Contributions}

We formally define the \emph{multiplayer Colonel Blotto game}, derive equilibria in several settings of the game, and provide linear time algorithms to sample from these equilibrium mixed strategies.
In multiplayer Blotto, there are $k \ge 2$ players with budgets $\B_1,\ldots,\B_k$, and, again, each battlefield is won by whichever player places the highest bid on it.
The game serves as a natural model for several of the famous applications studied in the two-player case, including the electoral competition application suggested by Myerson.

We focus on the symmetric case of multiplayer Blotto, where all players have the same budget, and construct efficiently-sampleable symmetric Nash equilibria for various settings of number of battlefields $n$ and number of players $k$:

\begin{enumerate}
\item 
We give equilibria for any number of players $k$ whenever the battlefields can be partitioned into $k$ sets of equal value (Theorem~\ref{thm:nequalsmk}). 
Furthermore, we provide an $O(n)$ time algorithm for sampling the randomized strategy (Algorithm~\ref{alg:nequalskcolonelBlottonashequilibrium}).\footnote{Throughout the paper, we use standard big-O notation $g(n) = O(f(n))$ to indicate that $\limsup_{n \to\infty} g(n) / f(n) \leq C$ for some constant $C$. }

\item 
We give equilibria for symmetric three-player Blotto whenever no battlefield accounts for more than one third of the value of all battlefields (Theorem~\ref{thm:kequalsthree}). 
We again provide an $O(n)$ algorithm to compute all these equilibrium strategies (Algorithm~\ref{alg:kequalsthree}).

This result is the highlight of our work.
The proof takes advantage of a connection between the Dirichlet distribution and the uniform distribution on the 2-sphere $\mathbb{S}_2$, as well as the rotational invariance of the uniform distribution on $\mathbb{S}_2$. The core technical challenge is constructing a suitably-structured map from $\mathbb{S}_2$ to $\R^n$, essentially converting a 3-battlefield equilibrium into an $n$-battlefield equilibrium.
\end{enumerate}

We also introduce a simple variant of Blotto which we call the \emph{Boolean Colonel Blotto game}. Boolean Blotto is the same as normal Blotto except players have integer budgets and their bids on each battlefield are restricted to be 0 or 1. 
In other words, players choose which subset of battlefields to {\em compete} on (i.e., bid 1 on).  The value of each battlefield is, as in Blotto, split evenly among the players who bid the most on it. 
In Section~\ref{sec:BooleanBlotto} we formally define and analyze this game in the multiplayer setting, which turns out to be significantly more interesting than the two-player Boolean setting. 
We give equilibria for all values of $k$ for Boolean Blotto regardless of battlefield valuations (Theorem~\ref{thm:BooleanBlotto}).
Interestingly, some of the quantities that arise in the equilibrium computation seem to be {\em hard} to compute, in the technical sense that it is not known how to compute them in polynomial time in the standard computing model.
Consequently, we are unable to give an {\em efficient} (polynomial-time) algorithm for players to sample from the exact Nash Equilibrium.
However, we derive a fully-polynomial-time approximation scheme for sampling the strategies (Algorithm~\ref{alg:BooleanvaluedcolonelBlotto}), i.e, an algorithm that efficiently samples a strategy from an $\eps$-approximate Nash Equilibrium for any given $\eps > 0$, however small it may be.
In particular, our algorithm runs in time polynomial in $n$, $k$, and $\log(1/\eps)$.

\subsection{Motivation}

In the century after its introduction by Borel, the Blotto game has seen a plethora of applications. Many of these naturally generalize to the multiplayer setting. Some are even more natural to consider with many players. Here are just a few examples:

\begin{description}
    \item  [Elections:]
    $k$ candidates or parties compete across $n$ winner-take-all districts \cite{myerson,laslierpicard,laslier,merolla}. $k=2$ corresponds to a two-party system, while $k\geq 3$ corresponds to a multi-party system. Each candidate or party must decide how to allocate campaign funds, or candidate time, across districts. One could also consider individual voters in a single-district election to be battlefields, as Myerson did in \cite{myerson}.
    
    \item [R\&D:] 
    $k$ companies have the ability to use their fixed R\&D budgets to research and develop $n$ potential drugs \cite{golmanpage, kvasov}. If the first company to develop the drug will receive the patent and all the profits for that drug, then this is a Blotto game.
    
    \item [Local Monopolies:]
    $k$ competing companies in the same industry want to become the dominant player in each of $n$ new local markets. If each market will tend to be dominated by the company that allocates the most resources to the market (due to network effects, for example) then this is a Blotto game.
    
    \item [Advertising:]
    $k$ companies compete to advertise a substitute good to $n$ consumers \cite{friedman}. Each consumer will probably only purchase one of the substitutes, so each battlefield (consumer) is indeed winner-take-all.
    
    \item [Ecology:]
    $k$ species in a habitat compete to fill $n$ distinct ecological niches \cite{golmanpage}. In this setting, if each niche can only be filled by one species, we can potentially think of the species as evolving Blotto strategies through natural selection.
\end{description}

There are also substantial mathematical connections between Blotto and simultaneous all-pay auctions \cite{roberson, robersonkvasov}. It is natural to consider these in the multiplayer setting.


Boolean Blotto, on the other hand, is a good model for any Blotto-type situation where whether to compete in a battlefield is a binary decision. For example, consider an election---perhaps a local election, or party primary---in which there are $n$ issues and the $k$ candidates distinguish themselves by choosing some subset of issues to focus on. Or consider $k$ companies each marketing substitute products (e.g. medications) by highlighting certain features. Finally, one could consider any setting in which $k$ people must each decide which of $n$ games of chance to compete in (at no cost). Beyond its immediate applications, we introduce Boolean Blotto because it is a simple variation of the standard Blotto game that requires completely different mathematical techniques to analyze.

\subsection{Proof overview}
\paragraph{Proof overview}
In order to derive the mixed Nash equilibria of Theorems~\ref{thm:nequalsmk},~\ref{thm:kequalsthree} and~\ref{thm:BooleanBlotto} for the multiplayer Blotto games, we  construct equilibria for the General Lotto version of the game, and then show a coupling of each player's bid distributions into a joint distribution that satisfies the budget constraint. 
Solving for Lotto equilibria is easier, because it allows us to think of a player's bid distributions on different battlefields as independent marginal distributions, rather than as a $n$-dimensional joint distribution.
General Lotto was solved by Myerson \cite{myerson} in the symmetric homogeneous multiplayer setting.
We extend his techniques in Sections~\ref{sec:multiplayerBlotto} and \ref{sec:BooleanBlotto} in order to derive the unique symmetric equilibria for the symmetric heterogeneous multiplayer setting and the symmetric heterogeneous Boolean-valued multiplayer setting.

Following an approach similar to \cite{hart}, we use our solutions to General Lotto to derive Lemmas~\ref{lem:colonelBlottosuffconds} and~\ref{lem:suffcondsboolvalued}, which are sufficient conditions for Colonel Blotto players to be in an equilibrium. 
These sufficient conditions for Blotto equilibria apply in some generality, and may be useful in the future for extending our Colonel Blotto results. 
The sufficient conditions reduce solving Colonel Blotto to the problem of constructing a joint distribution of bids with marginal bid distributions corresponding to a General Lotto equilibrium, subject to the constraint that the sum of each player's bids is almost surely equal to the player's budget.

For each of our three main theorems, we show the existence of the desired couplings constructively, by directly giving efficient linear-time algorithms to sample from the coupled distributions.
Each algorithm uses a different technique to couple the given marginal distributions. 
In order to prove Theorem~\ref{thm:nequalsmk}, we use special properties of the Dirichlet distribution. 
The crux of Theorem~\ref{thm:kequalsthree}'s proof involves efficiently transforming a list of battlefield valuations $(v_1,\ldots,v_n)$ into a corresponding matrix that rotates the 2-sphere about the origin in $n$-dimensional hyperspace.
Interestingly, we show that sampling from the surface of this rotated sphere and returning the squares of the coordinates yields a sample from a properly coupled distribution.
An interesting characteristic of this proof method is that the existence of a distribution (that couples the Lotto marginals) is established via an efficient sampling procedure,
in contrast to the typical approach of finding an efficient sampling procedure for a known distribution. 
Finally, in order to prove Theorem~\ref{thm:BooleanBlotto}, we use a greedy construction that couples arbitrary Bernoulli random variables subject to a budget constraint.



\subsection{Qualitative discussion of theorems} 

A principal goal of analyzing the Blotto game is to help applied researchers understand the qualitative differences that arise as the number of players or battlefields changes.
We interpret these limiting behaviors obtained from our derivations here.

For standard multiplayer Blotto (by Theorems~\ref{thm:nequalsmk} and~\ref{thm:kequalsthree}):
\begin{enumerate}
    \item 
    For a fixed number of players $k$, as the number of equally-valued battlefields increases, each player's bids become more evenly spread out across the battlefields, tending to the uniform distribution.
    \item
    When the number of players equals the number of battlefields, then the players play much higher bids on some battlefields than on others.
\end{enumerate}

In the Boolean-valued case (by Theorem~\ref{thm:BooleanBlotto})
\begin{enumerate}
    \item 
    Each player $i$ places a bid on each battlefield $j$ with some probability $p_j$. As the number of players $k$ tends to infinity, each equilibrium bid probability $p_j$ tends to $(v_j/V)\cB$, roughly speaking. (See Remark~\ref{rem:limitlargekbool}.)
    \item
    Similarly to the continuous-valued case, as the number of players increases, the bid probabilities become more spread out, in the sense that each player is more likely to compete in less valuable battlefields. (See Remark~\ref{rem:increaseingkstrategychangebool}.)
\end{enumerate}

For all three theorems above, we show how to sample efficiently from the joint distribution that we construct. In this sense, the Nash equilibria that we derive can be efficiently implemented in practice.

\subsection{Prior work}\label{subsec:priorwork}

\begin{table}[h]\label{tab:literature}\small
\rowcolors{1}{}{lightgray}
\begin{tabularx}{\textwidth}{m{6.5cm}YYYY}\hline
	
	& Asymmetric budgets & Heterogeneous values & $n>3$ battlefields & Number of players ($k$)\\ \hline
	
	Borel \& Ville \cite{borelville} &  &  &  & 2  \\
	
	Gross \& Wagner result 2 \cite{grosswagner} &  & \checkmark &  & 2  \\
	
	Gross \& Wagner result 3 \cite{grosswagner} &  &  & \checkmark & 2  \\
	
	Gross \cite{gross} &  & \checkmark & \checkmark & 2  \\
	
	Laslier \cite{laslier} &  & \checkmark & \checkmark & 2  \\
	
	Roberson \cite{roberson} & \checkmark &  & \checkmark & 2  \\
	
	Schwartz et al. \cite{schwartz} (partial result) & \checkmark & \checkmark & \checkmark & 2  \\
	
	Kovenock \& Roberson \cite{kovenock} (partial result) & \checkmark & \checkmark & \checkmark & 2  \\
	
	\textbf{Theorem~\ref{thm:nequalsmk} (partial result)} &  & \pmb\checkmark & \pmb\checkmark & $\mathbf{\geq 3}$ \\
	
	\textbf{Theorem~\ref{thm:kequalsthree} (handles most cases)} &  & \pmb\checkmark & \pmb\checkmark & \textbf{3} \\
	\hline
	
\end{tabularx}
\caption{Summary of equilibrium constructions in the continuous Blotto literature. Check marks indicate whether the result handles the condition in the column heading. Notable omissions: Myerson's multiplayer construction \cite{myerson}, which only provides equilibria for the Lotto game; results for the discrete setting; and our Theorem~\ref{thm:BooleanBlotto}, which solves our multiplayer Boolean Blotto setting.}
\end{table}

The Colonel Blotto game has been the subject of a considerable body of work over the course of a century. The game (both the discrete budget and continuous budget variations) was first introduced, without a general solution, by Émile Borel in 1921 \cite{borel}. This paper was, notably, the first ever in the nascent game theory literature to describe the concepts of pure and mixed strategies. Borel referred to Blotto as among the simplest games ``for which the manners of playing form a doubly infinite continuum.'' In 1938, Borel and Ville \cite{borelville} found equilibria for \emph{symmetric} \emph{homogeneous} three-battlefield Blotto. In a pair of papers in 1950, Gross and Wagner \cite{gross, grosswagner} found equilibria for all $n$, including for the heterogeneous setting. During the postwar period, there was a sizable classified military literature on Blotto in the United States \cite{blackett}.

In 1981, with applications to financial investment in mind, Bell and Cover introduced what, in modern terminology, would be called the one-battlefield \emph{General Lotto game} \cite{bellcover}. Myerson, apparently independently, described one-battlefield General Lotto in 1993, in the context of political economy \cite{myerson}. Myerson's paper is very relevant to our work because it appears to be the only prior work that considers generalizing Blotto (or rather, the easier-to-analyze Lotto) to a multiplayer setting. Myerson considers an infinite family of multiplayer generalizations corresponding to different voting systems; the natural multiplayer game we consider corresponds in this taxonomy to the plurality voting system. Myerson derived the unique symmetric Nash equilibria for these multiplayer General Lotto games; these correspond to the marginals of equilibrium strategies in our setting, as in Lemma~\ref{lem:colonelBlottosuffconds}. Note that Myerson dealt with General Lotto rather than Colonel Blotto precisely because it is easier to deal with: ``The advantage of my simplified formulation is that it will enable us to go beyond this `Colonel Blotto' literature and get results about more complicated situations in which more than two candidates are competing.'' In our paper, we obtain results in these complicated situations in the rich regime of Blotto.

The current century has seen a resurgence of interest in the Blotto game \cite{laslierpicard, laslier}.
In a landmark 2006 paper, Roberson found equilibria for all $n$ for the homogeneous non-symmetric setting \cite{roberson}. A string of recent works has worked towards the still incomplete goal of characterizing solutions to Blotto in the heterogeneous non-symmetric setting \cite{schwartz, kovenock, thomas}. Kovenock and Roberson's paper \cite{kovenock} includes a survey of progress on this question.

A simultaneous recent line of work has dealt with the \emph{discrete} version of Colonel Blotto, in which players' budgets are composed of indivisible units (i.e., their bids must lie in $\ZZ_{\geq 0}$). Our Boolean Blotto game can be thought of as a restricted version of discrete Blotto in which bids must lie in $\{0,1\}$. In 2008, Hart solved homogeneous symmetric discrete Blotto, and gave the General Lotto game its name \cite{hart}. In 2012, Dziubi\'nski solved non-symmetric discrete General Lotto \cite{dziubinski2013non}. Also in the discrete setting, Hortala-Vallve and Llorente-Saguer introduced a variant of Blotto in which the two players can value battlefields differently, and identified some pure strategy equilibria for this case \cite{hortala}. Many other variations of Blotto have been introduced over the decades in both the continuous and discrete settings \cite{tukey, shubikweber, kvasov, golmanpage}.

Several recent papers have given algorithms for variations of the discrete Blotto game, typically in time polynomial in the number of battlefields $n$ and the size of players' budgets \cite{ahmadinejad, behnezhad, behnezhadblum, behnezhadblum2}. Our algorithms for the continuous and Boolean-valued settings, in contrast, have running time polynomial in $n$ and the logarithm of the budget size.

Another side of the Blotto literature applies the Blotto model to various social science settings. In addition to the early military applications and later political economy and finance applications, Blotto has also been used to study topics such as U.S. presidential elections \cite{merolla}, terrorism \cite{arce}, phishing \cite{chia}, and advertising \cite{friedman}. It is closely related to the study of all-pay auctions \cite{baye}. Still another line of work, experimental in nature, tries to determine what strategies people will actually use in real-life Blotto games---see \cite{dechenaux} for a survey.

\subsection{Organization of paper}

The remainder of the paper contains three sections.
Section~\ref{sec:multiplayerBlotto} formally defines and solves cases of the multiplayer continuous Blotto game, and
Section~\ref{sec:BooleanBlotto} formally defines and solves the multiplayer Boolean Blotto game.
Finally, we end with some remarks and open problems in Section~\ref{sec:remarks}.

\section{Colonel Blotto equilibria}
\label{sec:multiplayerBlotto}

In this section, we formally define the General Lotto and Colonel Blotto games for multiple players, solve for their equilibria and construct efficient sampling methods for the equilibrium strategies.
The Colonel Blotto equilibria are presented in Theorems~\ref{thm:nequalsmk} and Theorem~\ref{thm:kequalsthree}. 
We begin by formally defining the Blotto game.
\begin{definition}
\label{def:multiplayer-Blotto}
The multiplayer {\em Colonel Blotto} game is specified by a tuple 

$$\left(k \in \NN, n \in \NN, \vec{\B} \in \RR_{\geq 0}^n, \vec{v} \in \RR_{\geq 0}^n \right),$$ where $k$ is the number of players, $n$ is the number of battlefields,
$\B_i$ is the budget of player $i \in [k]$, and
$v_j$ is the value of the battlefield $j \in [n]$.
We denote the sum total of the battlefield values by $V = \lone{\vec{v}} = \sum_{j=1}^n v_j$. 

Each player $i \in [k]$ plays a bid vector $A_{i,*} = (A_{i,1},\ldots,A_{i,n}) \in \RR_{\geq 0}^n$ satisfying the budget constraint
$$\lone{A_{i,*}} = \sum_{j \in [n]} A_{i,j} \leq \B_i.$$  Let the {\em bid matrix} $A = (A_{i,j})_{(i,j) \in [k] \times [n]}$ be the matrix whose $i$th row is $A_{i,*}$. For each $i \in [k]$, the payoff for player $i$ is
$$U_i(A) := \sum_{j \in [n]} U_{i,j}(A) := \sum_{j \in [n]} v_{j} \cdot \left(\frac{\mathbbm{1}(i \in \arg\max_{i' \in [k]} A_{i',j})}{|\arg \max_{i'} A_{i',j}|}\right).$$
In words: each battlefield's value is split evenly among the players who tied for the highest bid on that battlefield. The game is called {\em symmetric} if all the player budgets are equal, and {\em homogeneous} if all battlefield values are equal.
\end{definition}

A result of Dasgupta and Maskin establishes the existence of Nash equilibria for all values of $k$ and $n$, and guarantees the existence of symmetric equilibria in the symmetric-budget setting \cite{dasgupta1986existence}. In this paper we give explicit symmetric equilibria for the symmetric setting. Our first theorem holds for any number of players, but restricted battlefield values:

\begin{restatable}{theorem}{nequalsmk} \label{thm:nequalsmk}

Suppose that in the Colonel Blotto game with equal budgets $\cB_i = 1$, we are given a $k$-partition $\pi : [n] \to [k]$ of the battlefields such that there is equal value on each set of the partition:

$$\sum_{l \in \pi^{-1}(m)} v_l = \frac{1}{k} \sum_{l=1}^n v_l = \frac{V}{k} \quad \forall m \in [k].$$

Then if each of the players independently runs Algorithm~\ref{alg:nequalskcolonelBlottonashequilibrium}, the players will be in Nash equilibrium. Moreover, Algorithm~\ref{alg:nequalskcolonelBlottonashequilibrium} runs in $O(n)$ time.
\end{restatable}

The following important special case of this theorem immediately follows by defining $\pi(j) := (j \mod k) + 1$.
\begin{corollary} \label{cor:equalbudgets}
Suppose that in the Colonel Blotto game with equal budgets $\cB_i = 1$ there are $n = mk$ battlefields of equal value $v_i = V/n$, for some $m \in \N$. Then Algorithm~\ref{alg:nequalskcolonelBlottonashequilibrium} gives a Nash equilibrium in $O(n)$ time.
\end{corollary}

Our second main theorem holds only for three player games ($k = 3$ case), but allows us to handle a much wider range of battlefield valuations:

\begin{restatable}{theorem}{kequalsthree} \label{thm:kequalsthree}
Suppose that in the $3$-player Colonel Blotto game with equal budgets $\cB_i = 1$, the valuations satisfy $$v_j \leq \frac{V}{3}, \quad \forall j \in [n].$$
Then if each of the players independently runs \texttt{SampleBid}$(\vec{v})$ from Algorithm~\ref{alg:kequalsthree}, the players will be in Nash equilibrium. Moreover, Algorithm~\ref{alg:kequalsthree} runs in $O(n)$ time.
\end{restatable}


The main difficulty in proving Theorems~\ref{thm:nequalsmk} and~\ref{thm:kequalsthree} is that the strict budget constraint of the Blotto game generally means that a given player's bid distributions on the various battlefields have to be correlated, so that any bid-vector sampled from this joint distribution sums to one.
That is, each player $i$'s bids must be coupled in some potentially complicated way so that the budget constraint $\sum_{j=1}^n A_{i,j} \leq \cB_i$ holds with probability 1 over player $i$'s mixed strategy. 
In order to overcome this difficulty, we follow the meta-approach of \cite{hart} and prove both theorems by first analyzing the simpler General Lotto game. 
This is a variant of the Colonel Blotto game in which the budget constraints are relaxed to hold only in expectation over each player's bids, instead of almost surely:
\begin{definition}\label{def:multiplayer-lotto}
An instance of the {\em General Lotto} game is specified by a tuple $(k,n,\vec{\cB},\vec{v})$, as in the Colonel Blotto game. However, instead of playing a real-valued bid for each battlefield, each player plays a distribution of bids. For each $i \in [k]$ and $j \in [n]$, player $i$ plays a distribution $\cD_{i,j}$ over $ \RR_{\geq 0}$ such that the budget constraint is met in expectation:
$$\sum_{j=1}^n \EE_{A_{i,j} \sim \cD_{i,j}}[A_{i,j}] \leq \cB_i.$$
The payoff function for player $i \in [k]$ given the bids of all the players is $\EE_A U_i(A),$ where for each $i' \in [k],j' \in [k]$ the bids $A_{i',j'} \sim \cD_{i',j'}$ are drawn independently.
\end{definition}
Given a Nash equilibrium $(\cD_{i,j})_{i \in [k], j\in [n]}$ of the General Lotto problem, our approach will be to try to convert $(\cD_{i,j})_{i,j}$ into a Nash equilibrium of the Colonel Blotto problem. Our objective will be to construct a random variable $A \in \RR_{\geq 0}^{k \times n}$ such that the rows $A_{i,*}$ are independent of each other, such that $A_{i,j} \sim \cD_{i,j}$ for each $i \in [k], j \in [n]$, and such that the budget constraint $\lone{A_{i,*}} \leq \cB_i$ holds for each $i \in [k]$ almost surely. These conditions will ensure that $A$ is a mixed Nash equilibrium for the Colonel Blotto problem. We realize this program as follows: in Section~\ref{subsec:contgenlottoequil}, we characterize symmetric General Lotto equilibria, in Section~\ref{subsec:contsuffcond} we derive a sufficient condition for symmetric Colonel Blotto equilibria, and in Sections~\ref{subsec:thmnequalsmk} and~\ref{subsec:thmkequals3} we use this sufficient condition to prove Theorems~\ref{thm:nequalsmk} and~\ref{thm:kequalsthree}.

\subsection{General Lotto equilibria} \label{subsec:contgenlottoequil}

We now construct symmetric multiplayer General Lotto equilibria. Our construction is similar to Myerson \cite{myerson}, who constructed equilibria for the homogeneous case and proved that they were unique. Similar arguments to \cite{myerson} would prove uniqueness of our construction in the heterogeneous case, but for the sake of brevity we omit these arguments since they are not necessary in order to obtain sufficient conditions for Colonel Blotto equilibria. First, recall the definition of the Beta distribution:
\begin{definition}
For any $\alpha,\beta > 0$, the $\mathrm{Beta}(\alpha,\beta)$ distribution is the distribution supported on the interval $[0,1]$ with PDF proportional to $x^{\alpha-1}(1-x)^{\beta-1}$. In particular, if $X \sim \mathrm{Beta}(\alpha,1)$, then the CDF is $\mathbb{P}[X \leq x] = x^{\alpha}$ for all $x \in [0,1]$.
\end{definition}

\begin{lemma} \label{lem:generallottocontinuousvalued}
Consider the (continuous-valued) symmetric multiplayer General Lotto game $(k, n, \vec{\mathcal{B}} = \vec{\mathbbm{1}}, \vec{v})$ with $k \geq 2$ players and equal budgets $\cB_i = 1$. Suppose that for each $i \in [k]$ and $j \in [n]$, player $i$ plays distribution $\cD_{i,j} = \frac{kv_j}{V} \cdot \mathrm{Beta}(\frac{1}{k-1},1)$ on battlefield $j$. Then the players are in Nash equilibrium.
\end{lemma}
\begin{proof}
For this proof, let $X \sim \mathrm{Beta}(1/(k-1),1)$. First, the General Lotto budget constraint is satisfied for all $i \in [k]$ $$\sum_{j=1}^n \EE_{A_{i,j} \sim \cD_{i,j}}[A_{i,j}] = \sum_{j=1}^n \frac{kv_j}{V}\E\left[X\right] = \sum_{j=1}^n \frac{v_j}{V} = 1 = \cB_i.$$

Now suppose that player $k$ deviates by playing distributions $\cD'_{k,1},\ldots,\cD'_{k,n}$ meeting the General Lotto budget constraint. For all $i \in [k-1], j \in [n]$ let $A_{i,j} \sim \cD_{i,j}$, and $A_{k,j} \sim \cD'_{k,j}$ be independent random variables.
The expected payoff of player $k$ from battlefield $j$ is 
\begin{align*}
\EE[U_{k,j}(A) \mid A_{k,*}] &= v_j \cdot \PP[\forall i \in [k-1], A_{i,j} \leq A_{k,j} \mid A_{k,j}] \\ 
 &= v_j \cdot \prod_{i=1}^{k-1}\PP[A_{i,j} \leq A_{k,j} \mid A_{k,j}] = v_j \cdot \left(\PP\left[\frac{k v_j}{V} \cdot X \leq A_{k,j}\right]\right)^{k-1} \\ 
 &= v_j \cdot \min\left(1, \left(\frac{V}{k v_j} \cdot A_{k,j}\right)^{1/(k-1)}\right)^{k-1} 
 = v_j \cdot \min\left(1,\frac{V}{k v_j} A_{k,j}\right) \leq \frac{V}{k}A_{k,j},
\end{align*} 
where we have used that ties are measure-zero events. Therefore, 
$$\EE[U_k(A)] = \sum_{j \in [n]} \EE[U_{k,j}(A)] \leq \frac{V}{k} \cdot \sum_{j \in [n]} \EE[A_{k,j}] \leq \frac{V}{k}$$ 
The last inequality is the General Lotto budget constraint. By symmetry between the players, if $\cD'_{k,j} = \cD_{k,j}$ for all $j \in [n]$ then this upper bound is achieved: 
$\EE[U_k(A)] = \frac{V}{k}$. So playing according to Lemma~\ref{lem:generallottocontinuousvalued} is indeed a Nash equilibrium.
\end{proof}

\subsection{Sufficient Conditions for Colonel Blotto equilibrium} \label{subsec:contsuffcond}
The General Lotto equilibria of Lemma~\ref{lem:generallottocontinuousvalued} immediately give sufficient conditions for players to be in Colonel Blotto equilibrium:
\begin{lemma}
\label{lem:colonelBlottosuffconds} 
Consider the symmetric Colonel Blotto game $(k, n, \vec{\mathcal{B}} = \vec{\mathbbm{1}}, \vec{v})$. 
The players are in equilibrium if each player $i \in [k]$ independently bids a random vector of bids $A_{i,*} = (A_{i,1},\ldots,A_{i,n})$ such that:
\begin{enumerate}
\centering
    \item[(a)] $\sum_{j=1}^n A_{i,j} \leq 1 = \mathcal{B}_i$. \qquad {\em (b)} $A_{i,j} \sim \frac{kv_j}{V} \cdot \mathrm{Beta}(1/(k-1), 1)$.
\end{enumerate}

\end{lemma}
\begin{proof}
The budget constraints are met by (a). By linearity of expectation, the utilities only depend on the marginal distributions of the players' bids for each battlefield. So, if any player deviates from the strategy, then by Lemma~\ref{lem:generallottocontinuousvalued} and the fact that any Colonel Blotto strategy is also a General Lotto strategy, the deviating player's utility cannot improve. 
\end{proof}

Therefore, we have reduced the problem of computing Colonel Blotto equilibria to the problem of coupling Beta-distributed variables so as to satisfy the budget constraint.
In the following two sections, we give computationally efficient constructions of such couplings in order to prove Theorems~\ref{thm:nequalsmk} and~\ref{thm:kequalsthree}.

\begin{remark}[Blotto $\neq$ Lotto] \label{rem:lottoneqBlotto}
The conditions in Lemma~\ref{lem:colonelBlottosuffconds} are not necessary for players to be in Blotto equilibrium. For example, in the Colonel Blotto game specified by $(k=2,n=1,\vec{\cB} = \mathbbm{1},\vec{v} = \mathbbm{1})$, Lemma~\ref{lem:colonelBlottosuffconds} would require the distribution $2 \cdot \mathrm{Beta}(1,1)$, which is equal to $\mathrm{Unif}[0,2]$, to have support $\leq 1$ in order to meet condition (a). Clearly this is not the case, so the conditions of Lemma~\ref{lem:colonelBlottosuffconds} are not satisfied, and yet the Colonel Blotto game still has an equilibrium (in which both players play all of their budget on the one battlefield).
\end{remark}

\subsection{Couplings for arbitrary numbers of players (Theorem~\ref{thm:nequalsmk})} \label{subsec:thmnequalsmk}
We now prove Theorem~\ref{thm:nequalsmk} using the sufficient condition of Lemma~\ref{lem:colonelBlottosuffconds}. We will make use of a property of 
the {\em multivariate Beta distribution}---also known as the {\em Dirichlet distribution}.

\begin{definition}
\label{def:dirichlet-distribution}
The {\em Dirichlet distribution} $\text{Dir}(\alpha_1,\ldots,\alpha_m)$ is the distribution on the $(m-1)$-simplex $\Delta_{m-1}$ with density function $f(\vec{x};\vec{\alpha}) \propto \prod_{i = 1}^{m} x_i^{\alpha_i - 1}$.
\end{definition}


\begin{proposition}[folklore, e.g. \cite{lin2016dirichlet}] \label{prop:dirichletproperties}
Let $(X_1,\ldots,X_m) \sim \mathrm{Dir}(\alpha_1,\ldots,\alpha_m)$. Then
\begin{enumerate}
\item[(i)] For each $i \in [m]$, $X_i \sim \mathrm{Beta}\left(\alpha_i, \sum_{j \ne i}\alpha_j \right)$.
\item[(ii)] $\sum_{i=1}^m X_i = 1$ almost surely
\end{enumerate}
\end{proposition}

Proposition~\ref{prop:dirichletproperties} implies that the Dirichlet distribution on $\Delta_{k-1}$ with parameters $\vec{\alpha} = \frac{1}{k-1}\vec{\mathbbm{1}}$ has marginals equal to $\mathrm{Beta}(1/(k-1), 1)$.
This leads us to the following algorithm to sample a symmetric Nash equilibrium strategy for each player in a $k$-player Blotto game where the battlefields can be partitioned into $k$ sets of equal value.

\begin{algorithm}[H]
\SetAlgoLined
\DontPrintSemicolon
\SetKwInOut{Input}{Input}
\SetKwInOut{Output}{Output}
\SetKwFunction{ConstructM}{ConstructM}
\KwIn{a Colonel Blotto game $(k,n,\vec{\cB} = \mathbbm{1},\vec{v})$ and a partition function $\pi : [n] \to [k]$ satisfying $\sum_{\ell \in \pi^{-1}(m)} v_{\ell} = V/k$ for each $m \in [k]$.}

\KwOut{a sample $(A_{1},\ldots,A_{n}) \in \RR^n$ from a mixed equilibrium strategy for a single player.}
\BlankLine

Draw $(X_1,\ldots,X_k) \sim \mathrm{Dir}(1/(k-1),\ldots, 1/(k-1))$.

$A_j \gets \left(\frac{k v_j}{V}\right) \cdot X_{\pi(j)}$ for all $j \in [n]$.

\Return $(A_1,\ldots,A_n)$
 \caption{\texttt{NashEquilThm\ref{thm:nequalsmk}}: Colonel Blotto Nash Equilibrium for Theorem~\ref{thm:nequalsmk}}
 \label{alg:nequalskcolonelBlottonashequilibrium}
\end{algorithm}

\begin{proof}[Proof of Theorem~\ref{thm:nequalsmk}]
\textit{Correctness}: The output of Algorithm~\ref{alg:nequalskcolonelBlottonashequilibrium} meets the conditions of Lemma~\ref{lem:colonelBlottosuffconds} and therefore the players are in Nash equilibrium: \begin{enumerate}
\item[(a)] Budget constraint:
$$\displaystyle\sum_{j=1}^n A_{i,j} = \displaystyle\sum_{j=1}^n \left(\frac{k v_j}{V}\right) X_{i,\pi(j)} = \displaystyle\sum_{m=1}^k X_{i,m} \left(\displaystyle\sum_{l \in \pi^{-1}(m)} \frac{k v_l}{V}\right) = \displaystyle\sum_{m=1}^k X_{i,m}$$ 
which is 1 by Proposition~\ref{prop:dirichletproperties}(ii).
\item[(b)] Marginal constraint: $A_{i,j} \sim \left(\frac{k v_j}{V}\right) \cdot \mathrm{Beta}(1/(k-1), 1)$ by Proposition~\ref{prop:dirichletproperties}(i). 
\end{enumerate}\textit{Running time}: we can sample the Dirichlet variable in $O(n)$ time, using the method of \cite{ahrens1974computer} to sample $n$ i.i.d variables $Y_i \sim \mathrm{Gamma}(1/(k-1),1)$ and letting $X_i = \frac{Y_i}{\sum_{l=1}^n Y_l}$ for all $i \in [n]$.

\end{proof}

\subsection{Couplings for 3 players (Theorem~\ref{thm:kequalsthree})} \label{subsec:thmkequals3}

We prove Theorem~\ref{thm:kequalsthree}, which vastly improves over Theorem~\ref{thm:nequalsmk} (from the previous section) in the $k=3$ case. The proof of this theorem is much more involved, and is inspired by the relationship between the Dirichlet distribution and the $L_p$-norm uniform distribution defined in \cite{cambanis1981theory}.

In particular, \cite{song1997lpnorm} proves that given $(U_1,\ldots,U_m)$ drawn from the $m$-dimensional $L_p$-norm uniform distribution, then $(|U_1|^p,\ldots,|U_m|^p)$ is distributed as $\mathrm{Dir}(1/p,\ldots,1/p)$.  Therefore, for the construction of Theorem~\ref{thm:nequalsmk}, in order to draw $(X_1,\ldots,X_k)$ from the $\mathrm{Dir}(1/(k-1),\ldots,1/(k-1))$ distribution, we could have set $(X_1,\ldots,X_k) = (|U_1|^{k-1},\ldots,|U_k|^{k-1})$ for $(U_1,\ldots,U_k)$ drawn from the $L_{k-1}$-norm uniform distribution. The $k = 3$ case is very special, because the $L_{k-1} = L_2$-norm uniform distribution is the uniform distribution on the unit $L_2$ sphere, and therefore it is rotationally symmetric. We will take advantage of the rotational symmetry of the uniform distribution on the $L_2$ sphere in order to handle a much wider range of battlefield valuations in Theorem~\ref{thm:kequalsthree} when $k = 3$. We summarize this intuition by stating the following remarkable geometric fact:

\begin{proposition}\label{prop:unif2spheredotprod}
Let $U \in \RR^3$ be a point drawn uniformly at random from the surface of the unit sphere $\sum_{l=1}^3 U_l^2 = 1$. Let $c \in \RR^3$. Then the inner product $c \cdot U$ is distributed as \begin{align*}&c \cdot U \sim \mathrm{Unif}[-\|c\|, \|c\|], & \mbox{and so} \qquad\quad
&(c \cdot U)^2 \sim \|c\|^2 \cdot \mathrm{Beta}(1/2, 1).\end{align*}
\end{proposition}
\begin{proof}
By the rotational symmetry of $U$, the inner product $\frac{c}{\|c\|} \cdot U$ is equal in distribution to $U_1$. Since $U_1 \sim \mathrm{Unif}[-1,1]$ (see e.g., Theorem 2.1 of \cite{song1997lpnorm}), $c \cdot U \sim \mathrm{Unif}[-\|c\|, \|c\|]$ follows.
Therefore the CDF of $\frac{(c \cdot U)^2}{\|c\|^2}$ is $\PP[\frac{(c \cdot U)^2}{\|c\|^2} \leq a] = \sqrt{a}$ for any $a \in [0,1]$.
This proves that $\frac{(c \cdot U)^2}{\|c\|^2} \sim \mathrm{Beta}(1/2, 1)$. 
\end{proof}

The analysis of Algorithm~\ref{alg:kequalsthree}, which constructs the equilibrium for Theorem~\ref{thm:kequalsthree}, will depend on this proposition. In short, the algorithm samples a vector $U \in \R^3$ uniformly from the unit sphere $\mathbb{S}_2 \subset \RR^3$. It then maps $U$ into $\R^n$ with a linear isometry described by a matrix $M$. Finally, it outputs the coordinate-wise square of this point. In order to ensure correctness, the algorithm must use an isometry $M$ that has squared row norms proportional to the battlefield valuations. Finding such an $M$ is the core technical challenge, and it is accomplished by the helper algorithm \texttt{ConstructM}, which constructs an $M$ that has the following guarantee (proof deferred):
\begin{claim} \label{claim:constructmcorrectness}
Given values $0 \leq s_1,\ldots,s_n \leq 1$ and $m \in \ZZ$ such that $\sum_{j=1}^n s_j = m$, the method \texttt{ConstructM} returns in $O(nm)$ time a matrix $M \in \RR^{n \times m}$ such that $M^TM = I_m$ and $\|M_{j,*}\|^2 = s_j$, for all $j \in [n]$. (Here $M_{j,*}$ denotes the $j$th row of $M$.)
\end{claim}

\begin{algorithm}
\caption{\texttt{NashEquilThm\ref{thm:kequalsthree}}: Colonel Blotto Nash Equilibrium for Theorem~\ref{thm:kequalsthree}}\label{alg:kequalsthree}

\SetAlgoLined
\DontPrintSemicolon
\SetKwInOut{Input}{Input}
\SetKwInOut{Output}{Output}
\SetKwFunction{ConstructM}{ConstructM}

\KwIn{The number of battlefields $n$ and the battlefield valuations $v_1 \ge v_2 \ge \cdots \ge v_n \ge 0$ such that $v_j \leq \frac{1}{3} V$ for all $j \in [n]$.
Recall that $V = \sum_{k=1}^n v_k$.}
\BlankLine
\KwOut{A bid vector $A = (A_1,\ldots,A_n)^T \in \R^n$ for the $n$ battlefields that is sampled from a distribution satisfying Lemma~\ref{lem:colonelBlottosuffconds} for the Blotto game $G = (3, n, \vec{\mathbbm{1}}, \vec{v})$.}
\BlankLine

\SetKwProg{Fn}{Function}{ is}{end}
\Fn{$\texttt{SampleBid}(\vec{v})$}{
Construct $M \in \R^{n \times 3}$ by running: $M \leftarrow \ConstructM \left(\left(\frac{3}{V}\right) \cdot \vec{v}; m = 3\right)$ \label{algstep:kequals3constructm}

Sample $U \in \RR^3$ uniformly at random from the unit $\ell_2$-sphere $\mathbb{S}_2 = \{x \mid \|x\|_2 = 1\}$. \label{algstep:kequals3sampleu}

\Return $A_j \gets (M_{j,*} \cdot U)^2$ for all $j \in [n]$. \label{algstep:kequals3return}
}

\SetKw{KwAnd}{and}

\nonl \hrulefill

\KwIn{Values $0 \le s_1, s_2, \ldots, s_n \le 1$ and $m \in \N$ such that $\sum_{j=1}^n s_j = m$.} 

\BlankLine

\KwOut{$M \in \RR^{n \times m}$ such that $M^TM = I_m$ and $\|M_{j,*}\|^2 = s_j$, for all $j \in [n]$.}

\BlankLine

\Fn{$\texttt{ConstructM}(\vec{s}, m)$}{

Permute the indices of $\vec{s}$ so that $s_{r} \ge s_{r'}$ for each $r \in [m]$ and $r' \in [n] \setminus [m]$. \label{algstep:constructmrowpermute}

Initialize $M \in \RR^{n \times m}$ as $M_{i,i} = 1$ for all $i \in [m]$, and $0$ everywhere else.\;

$j \leftarrow 1$, $l \leftarrow m+1$.\;

\While{$j \leq m$ \KwAnd $l \leq n$}{ \label{algstep:constructminvariantstep}

$w_1, w_2 \leftarrow \texttt{RotatePair}(u_1 = M_{j,*},u_2 = M_{l,*},t_1 = s_j, t_2 = s_l)$.\;

$M_{j,*} \leftarrow w_1$.
$M_{l,*} \leftarrow w_2$.\;

\lIf{$\|M_{j,*}\|^2 = s_j$}{$j \leftarrow j + 1$.}
\lIf{$\|M_{l,*}\|^2 = s_l$}{$l \leftarrow l + 1$.}
}
Undo the row permutation from step \ref{algstep:constructmrowpermute}. \label{algstep:constructmrowpermuteback}

\Return $M$.\;
}

\nonl \hrulefill

\KwIn{Vectors $u_1,u_2 \in \RR^m$, and targets $t_1,t_2 \in \RR$ such that $\|u_1\|^2 \geq t_1 \geq t_2 \geq \|u_2\|^2$ and $u_1 \cdot u_2 = 0$.}

\BlankLine

\KwOut{$w_1,w_2 \in \RR^m$ that are (i) supported on $\mathrm{supp}(u_1) \cup \mathrm{supp}(u_2)$ such that (ii) $W = (\begin{array}{cc}w_1 & w_2\end{array}) \in \RR^{m \times 2}$ and $U = (\begin{array}{cc}u_1 & u_2\end{array}) \in \RR^{m \times 2}$ satisfy $WW^T = UU^T$, (iii) $\|w_1\|^2 \geq t_1 \geq t_2 \geq \|w_2\|^2$ and (iv) there is $k \in [2]$ such that $\|w_k\|^2 = t_k$.}

\BlankLine

\Fn{$\texttt{RotatePair}(u_1,u_2,t_1,t_2)$}{
\lIf{$\|u_1\|^2 = \|u_2\|^2$}
{$a \gets 1$, $b \gets 0$}

\eIf{$\|u_1\|^2 - t_1 \geq t_2 - \|u_2\|^2$}
{$a \leftarrow \sqrt{\frac{\|u_1\|^2 - t_2}{\|u_1\|^2 - \|u_2\|^2}}$, 
 $b \leftarrow \sqrt{1-a^2}.$}
{$a \leftarrow \sqrt{\frac{t_1 - \|u_2\|^2}{\|u_1\|^2 - \|u_2\|^2}}$, 
 $b \leftarrow \sqrt{1-a^2}.$}

\Return $w_1 = a u_1 - b u_2$, $w_2 = b u_1 + a u_2$.\;
}
\end{algorithm}

Assuming Claim~\ref{claim:constructmcorrectness}, we prove the correctness of \texttt{NashEquilThm\ref{thm:kequalsthree}}:

\begin{proof}[Proof of Theorem~\ref{thm:kequalsthree}]
\textit{Correctness}: The inputs to \texttt{ConstructM} in step~\ref{algstep:kequals3constructm} of Algorithm~\ref{alg:kequalsthree} satisfy the prerequisites $0 \leq s_1,\ldots,s_n \leq 1$ and $\sum_{j=1}^n s_j = m = 3$. Therefore, the matrix $M \in \RR^{n \times 3}$ is guaranteed to have the following properties by Claim~\ref{claim:constructmcorrectness}:
$M^T M = I_3$  and $\|M_{j,*}\|^2 = \frac{3 v_j}{V}$ for all $j \in [n]$. So if each player $i \in [3]$ bids $(A_{i,1}, \ldots, A_{i,n})$ by independently running Algorithm~\ref{alg:kequalsthree} with a random sphere point $U_i \in \RR^3$, then the sufficient conditions of Lemma~\ref{lem:colonelBlottosuffconds} are met:
\begin{enumerate}
 
 \item[(a)] Budget constraint: $\sum_{j=1}^n A_{i,j} = \sum_{j=1}^n (M_{j,*} \cdot U_i)^2 = \|MU_i\|^2 = U_i^T M^T M U_i = \|U_i\|^2 = 1$, using $M^TM = I_3$ and the fact that $U_i$ is on the unit sphere.
 
 \item[(b)] Marginal constraint: $A_{i,j} = \left(M_{j,*} \cdot U_i\right)^2 \sim \|M_{j,*}\|^2 \cdot \mathrm{Beta}(1/2, 1) = \frac{3 v_j}{V} \cdot \mathrm{Beta}(1/2, 1)$ by Proposition~\ref{prop:unif2spheredotprod} and $\|M_{j,*}\|^2 = 3v_j / V$.
\end{enumerate}

\textit{Running time}: The call to \texttt{ConstructM} with $m = 3$ in step~\ref{algstep:kequals3constructm} takes $O(n)$ time by Claim~\ref{claim:constructmcorrectness}. Sampling $U \in \RR^3$ in step~\ref{algstep:kequals3sampleu} takes $O(1)$ time, for example using the algorithm of \cite{muller1959note}. And finally step~\ref{algstep:kequals3return} takes $6n$ multiplications and additions. So the total running time is $O(n)$.
 
\end{proof}

\subsection{Proof of Claim~\ref{claim:constructmcorrectness} (\texttt{ConstructM} correctness)}

The algorithm \texttt{ConstructM} greedily updates a matrix $M \in \RR^{n \times m}$ using the helper algorithm \texttt{RotatePair} until the desired properties of $M$ are achieved. $M$ is initialized to the matrix $(I_m, 0)^T$.
Each application of \texttt{RotatePair} applies a linear rotation transformation to a pair of rows from $M$ such that at least one of these rows becomes scaled correctly, while the column-orthogonality of $M$ is maintained. Assuming correctness of the \texttt{RotatePair} subroutine, which is proved in Claim~\ref{claim:rotatepaircorrectness} of the appendix, an invariant argument demonstrates that greedily applying \texttt{RotatePair} works:
\begin{proof}[Proof of Claim~\ref{claim:constructmcorrectness}]\textit{Correctness}:
We analyze the algorithm by proving several invariants on $M,j,l$. These hold at step~\ref{algstep:constructminvariantstep}.

\textit{Invariant 1}: $s_j \geq s_l$.

\textit{Invariant 2}: The columns of $M$ are orthonormal. $M^T M = I_m$.

\textit{Invariant 3}: For all $r \in \{j,\ldots,m\}$ and $r' \in \{l,\ldots,n\}$,
$\|M_{r,*}\|^2 \geq s_r \mbox{ and } \|M_{r',*}\|^2 \leq s_{r'}$.

\textit{Invariant 4}: For all $r \in \{j,\ldots,m\}$ and $r' \in \{l,\ldots,n\}$, $M_{r,*} \cdot M_{r',*} = 0$.


The invariants clearly hold when the algorithm first reaches step \ref{algstep:constructminvariantstep}. We prove that they are maintained on each iteration. 
Let $M,j,l$ be the states of the variables before running an iteration of the while loop, and $M',j',l'$ the states of the variables after. If $M,j,l$ respect the invariants, then the preconditions of \texttt{RotatePair} are met, because $\|M_{j,*}\|^2 \geq s_j \geq s_l \geq \|M_{l,*}\|^2$ by Invariants 1 and 3, and $M_{j,*}\cdot M_{l,*} = 0$ by Invariant 4.

\textit{Invariant 1}: This follows from the preprocessing in step \ref{algstep:constructmrowpermute}, because $j \in [m]$ and $l \in [n] \sm [m]$.

\textit{Invariant 2}: Notice that for all $a,b\in [m]$, \begin{align*}((M')^T M')_{ab} &= \sum_{c \in [n]} M'_{ca} M'_{cb} = (M'_{ja} M'_{jb} + M'_{la} M'_{lb} - M_{ja}M_{jb} - M_{la}M_{lb}) + (M^T M)_{ab}\end{align*}
So since $M^T M = I_m$ by Invariant 2, it suffices to show that $M'_{ja} M'_{jb} + M'_{la} M'_{lb} - M_{ja}M_{jb} - M_{la}M_{lb} = 0$ for all $a,b$. This is precisely the condition that $$\left(\begin{array}{cc} M'_{j,*} & M'_{l,*}\end{array}\right)\left(\begin{array}{cc} M'_{j,*} & M'_{l,*}\end{array}\right)^T = \left(\begin{array}{cc} M_{j,*} & M_{l,*}\end{array}\right)\left(\begin{array}{cc} M_{j,*} & M_{l,*}\end{array}\right)^T,$$ which is guaranteed by item (ii) of \texttt{RotatePair}.

\textit{Invariant 3}: Since $j$ and $l$ are the only rows modified from the previous step, and $j' \geq j$, $l' \geq l$, it suffices to consider rows $j$ and $l$. For these, item (iii) of \texttt{RotatePair} guarantees that $\|M'_{j,*}\|^2 \geq s_j$ and $s_l \geq \|M'_{l,*}\|^2$. 

\textit{Invariant 4}: Item (iv) of \texttt{RotatePair} guarantees that at least one of $j$ and $l$ is incremented on each step. If $j' > j$, then the invariant holds, because the vectors $M'_{r,*}$ for $r \geq j'$ are supported on coordinates $\{r,\ldots,m\} \subset \{j',\ldots,m\}$, while by item (i) of \texttt{RotatePair} the vectors $M'_{r',*}$ for $r' \geq l' \geq l$ are supported on coordinates $\{1,\ldots,j'-1\}$. Otherwise, if $j' = j$ then $l' > l$, and the vectors $M'_{r',*}$ for $r' \geq l'$ are all 0. So in both cases $M'_{r,*} \cdot M'_{r',*} = 0$ for all $r \in \{i',\ldots,m\}$ and $r' \in \{j',\ldots,n\}$.

Therefore Invariants 1 through 4 are maintained by the algorithm. Notice that the row index $j$ (respectively, $l$) is only incremented if $\|M_{j,*}\|^2 = s_j$ (respectively, $\|M_{l,*}\|^2 = s_l$), and after that the row is no longer modified. So if the algorithm ever exits, then $\|M_{r,*}\|^2 = s_r$ for all $r \in \{1,\ldots,j-1\} \cup \{m+1,\ldots,l-1\}$. Now, if the algorithm exits, then $j = m+1$ and/or $l = n+1$. If $j = m+1$, we have by Invariant 2
\begin{align*}
m &= \trace(M^T M) = \sum_{r \in [n]} \|M_{r,*}\|^2 = \sum_{r \in [l-1]} s_r + \sum_{r \in [n] \sm [l-1]} \|M_{r,*}\|^2, & \mbox{since }j = m+1 \\ 
&\leq \sum_{r \in [n]} s_r = m, &\mbox{by Invariant 3}.
\end{align*}

If there were $k \in [n] / [j-1]$ such that $\|M_{k,*}\|^2 < s_k$ then the inequality in the last line would be strict. So we may conclude that $\|M_{k,*}\|^2 = s_k$ for all $k \in [n]$. Combining this with the knowledge that $M^T M = I_m$ by Invariant 2, we have shown that if $i$ ever reaches $n+1$, then the output is correct. Similarly, if $j$ ever reaches $m+1$, we may also argue that the output is correct. So it suffices to prove that the program terminates. This is true because item (iv) of \texttt{RotatePair} guarantees that either $i$ or $j$ is incremented on each step, and so the loop terminates after at most $n$ iterations.

\textit{Running time}: The initialization steps (including the permutation of the rows in steps \ref{algstep:constructmrowpermute} and \ref{algstep:constructmrowpermuteback}) take $O(mn)$ time, and each of the $\leq n$ iterations of the loop takes $O(m)$ time (because \texttt{RotatePair} takes $O(m)$ time). So the algorithm runs in $O(mn)$ total time.

\end{proof}

\section{Boolean-valued Colonel Blotto game}
\label{sec:BooleanBlotto}

We now turn our focus to analyzing the Boolean Blotto game. 
In this game, each player $i$ chooses whether to compete or not compete in up to $\mathcal{B}_i$ battlefields, and the values of battlefields are split evenly among the players who compete in them (or evenly among all players if nobody competes).

\begin{definition}\label{def:Boolean-Blotto}
The {\em Boolean-valued Colonel Blotto game} has the same payoff function as the continuous-valued Colonel Blotto game, with two additional restrictions:\\
(integer budget) each player $i \in [k]$ has an integer-valued budget $\cB_i \in \{0, \ldots, n\}$\\
(Boolean bids) each bid $A_{i,j}$ is either 0 or 1; we say player $i$ {\em competes} in battlefield $j$ if $A_{i,j} = 1$. \\
The game is {\em symmetric} if all players have the same budget $\B$.
\end{definition}

\begin{definition}\label{def:Boolean-multiplayer-lotto}
In the {\em Boolean-valued General Lotto} game, each player $i \in [k]$ plays a vector of probabilities $(p_{i,1},\ldots,p_{i,n})$, such that the budget constraint is met in expectation: $\sum_{j=1}^n p_{i,j} \leq \cB_i$. The payoff function for player $i$ given the bids of all the players is $\EE_A U_i(A),$ where for each $i' \in [k],j' \in [k]$ the bids $A_{i',j'} \sim \Ber(p_{i',j'})$ are drawn independently.
\end{definition}

\begin{lemma}
When there are $k=2$ players, it is a maximin pure strategy for player $i$ to compete only in the $\B_i$ battlefields of highest value.
\end{lemma}
\begin{proof}
Regardless of the other player's strategy, the marginal gain from competing in battlefield $j$ is $v_j/2$, so it is optimal to compete in the most valuable battlefields.
\end{proof}

Boolean Blotto only becomes interesting when $k > 2$. We now proceed to characterize the equilibria of symmetric multiplayer Boolean Blotto.

\subsection{Boolean General Lotto and sufficient conditions for Colonel Blotto}
As in our analysis of continuous-valued Colonel Blotto, we first characterize the symmetric equilibria of the General Lotto analogue of the game.

For a given player, Alice, and given battlefield of value $v$, let $u_1(p,v)$ be the expected utility earned by Alice from competing in the battlefield if all the other $k-1$ players independently compete with probability $p$, and let $u_0(p,v)$ be Alice's expected utility from not competing. Let $m_v(p) = u_1(p,v) - u_0(p,v)$ be the marginal utility of competing. We can write Alice's expected utility from competing with probability $q$ as $ q u_1 + (1-q)u_0$.

If Alice doesn't compete in the battlefield, she only gains utility when nobody competes: $u_0(p,v) = \frac{v}{k}(1-p)^{k-1}$. The total utility earned by all players from the battlefield is $v$, so by symmetry, $p \cdot u_1(p,v) + (1-p) \cdot u_0(p,v) = \frac{v}{k}$. Combining these two equations yields
$$m_v(p) = \begin{cases}
\frac{v}{k}(k-1), & p = 0 \\
\frac{v}{k} \cdot \frac{1-(1-p)^{k-1}}{p}, & 0 < p \leq 1 \end{cases}$$

We will show that there is a unique symmetric equilibrium. The probabilities $p_1, \dots, p_k$ of the equilibrium strategy are such that the marginal utilities of competing are essentially the same for all battlefields. This means that if all players including Alice play the equilibrium strategy, then Alice will have no incentive to move $\epsilon$ probability mass from one battlefield to another. To obtain these probabilities it will be useful to define an inverse of $m_v(p)$.

\begin{claim}\label{claim:continuous}
	When $k>2$, $m_v(p)$ is continuous and monotonically decreasing on the interval $p \in [0,1]$. Therefore it maps $[0,1]$ bijectively to $[v/k,(k-1)\cdot v/k]$.
\end{claim}
The proof is in Appendix~\ref{sec:boolproofs}. By Claim~\ref{claim:continuous}, the inverse $m_v^{-1}$ is uniquely defined on $[v/k,(k-1)\cdot v/k]$, and we may extend its domain to $\R$ by letting $m_v^{-1}(x) = 1$ for $x < v/k$ and $m_v^{-1}(x) = 0$ for $x > (k-1)\cdot v/k$. We are now ready to characterize the symmetric General Lotto equilibrium:

\begin{lemma} \label{lem:generallottoboolvalued}
	The following is the unique symmetric Nash equilibrium of the Boolean-valued General Lotto game with $k \geq 3$ players, equal integer-valued budgets $\cB_i = \cB$ and battlefield valuations $v_1 \geq v_2 \geq \dots \geq v_n > 0$: 
	\begin{equation}
	p_j = m_{v_j}^{-1}(x^*) \quad \forall j \in [n]
	\end{equation}
	where $x^* = \inf \left\{x \in \R: \sum_{j=1}^n m_{v_j}^{-1}(x) \leq \cB \right\}$.
\end{lemma}
\begin{proof}
	If $\cB = n$, then $x^* = -\infty$ so every $p_j = 1$; this is clearly the unique equilibrium. So, we assume that $\cB < n$ henceforth. Note that we chose $x^*$ such that the players meet their Lotto budget constraint exactly.
	
	Now suppose Alice deviates from the strategy by playing $\{q_j\}_{j \in n}$ while all other players play $\{p_j\}_{j \in n}$. The utility she gains by deviating is
	\begin{align}
	\sum_{j=1}^n (q_j - p_j)\cdot m_{v_j}(p_j) &= \sum_{j=1}^n (q_j - p_j)\cdot  m_{v_j}(m_{v_j}^{-1}(x^*)) \\
	&\leq \sum_{j=1}^n (q_j - p_j) x^* = x^*\left( \sum_{j=1}^n q_j -  \sum_{j=1}^n p_j \right) \label{line:invert}\\
	&= x^*(\cB - \cB) = 0
	\end{align}
	where the inequality in line~\eqref{line:invert} arises because $x^*$ may lie in the extended domain of $m_{v_j}^{-1}$ for some $j$s. 
	Thus, Alice has no incentive to deviate, so this is an equilibrium.
	
	Now we prove uniqueness. 
	Let $\{\pi_j\}_{j \in n}$ be any symmetric equilibrium. 
	We will show that there must exist an $x'$ such that $\pi_j = m_{v_j}^{-1}(x')$ for each $j$. 
	Assume for contradiction that 
	
    \begin{equation}\mbox{There is no }x'\mbox{ such that }\pi_j = m_{v_j}^{-1}(x')\mbox{ for each }j\mbox{.}\tag{$\ast$}\end{equation}
	
    We will show in cases that there are battlefields $i$ and $\ell$ such that: 
    (a) the marginal utility $m_{v_i}(\pi_i) < m_{v_\ell}(\pi_\ell)$, and 
    (b) probability $\pi_i > 0$ and probability $\pi_\ell < 1$.
    Thus, a player Alice will increase her utility by shifting $\eps$ probability mass from battlefield $i$ to $\ell$.

	\paragraph{Case 1} 
	Suppose $\pi_i \in (0,1)$ for some $i$. 
	Let $x' = m_{v_i}(\pi_i)$. 
	Assumption ($\ast$) implies that there is some $\ell \in [n]$ such that $\pi_\ell \ne m_{v_\ell}^{-1}(x')$.
	Three subcases ensue:
	(a) if $\pi_\ell = 0$, then $0 < m_{v_\ell}^{-1}(x')$, so applying the monotonically {\em decreasing} function $m_{v_\ell}$ to both sides of the inequality yields $m_{v_\ell}(\pi_\ell) = m_{v_\ell}(0) > x' = m_{v_i}(\pi_i)$.
	(b) if $\pi_\ell = 1$, then $1 > m_{v_\ell}^{-1}(x')$, so applying the monotonically {\em decreasing} function $m_{v_\ell}$ to both sides of the inequality yields $m_{v_\ell}(\pi_\ell) = m_{v_\ell}(1) < x' = m_{v_i}(\pi_i)$.
	(c) if $\pi_\ell \in (0,1)$, then either $m_{v_\ell}(\pi_\ell) > m_{v_i}(\pi_i)$ or $m_{v_\ell}(\pi_\ell) < m_{v_i}(\pi_i)$.

	\paragraph{Case 2} 
	Suppose $\pi_j \in \{0,1\}$ for all $j \in [n]$, yet there is no $x'$ such that $x' = m_{v_j}(\pi_j)$ for all $j \in [n]$.
	Then by Assumption~$(\ast)$ there must be indices $i, \ell \in [n]$ such that $\pi_i = 1$ and $\pi_\ell =0$ and $v_i < (k-1)v_\ell$ so $m_{v_i}(\pi_i) = m_{v_i}(1) = \frac{v_i}{k} < \frac{k-1}{k}v_\ell = m_{v_\ell}(0) = m_{v_\ell}(\pi_\ell)$.

	In all cases, moving $\eps$ probability mass from the battlefield with the smaller marginal utility to the larger (between $i$ and $j$) strictly increases Alice's utility and shows the $(\pi_1,\ldots,\pi_n)$ is not an equilibrium. 
	This contradicts Assumption~$(\ast)$, so there is an $x'$ such that $\pi_j = m_{v_j}^{-1}(x')$ for each $j$. It is an immediate consequence of the tightness of the budget constraint that it must be $x^*$, thereby completing the proof.
\end{proof}

\begin{remark}[Limit of large $k$] \label{rem:limitlargekbool}
Let us study the asymptotic behavior of the solution as the number of players $k$ tends to infinity and the average utility per player stays constant (so we increase the values $v_j$ proportionally with $k$). Notice that $m_{k\cdot v_j}(p)$ tends towards $v_j/p$ for each $j$, so the inverse $m_{k \cdot v_j}^{-1}(x)$ tends towards $\min(1,v_j/x)$ for $x > 0$.
Therefore, in the limit of large $k$, the equilibrium strategy tends towards surely competing in some of the top-valued battlefields and competing in the rest with probabilities proportional to the values of those battlefields.
Quantitatively, Lemma~\ref{lem:generallottoboolvalued} prescribes this strategy: iteratively assign portions of the budget to battlefields $1,\ldots,n$ in order of decreasing value as follows. Write $\cB^{(l)}$ to denote the budget remaining after assigning budget to battlefields $1,\ldots,l$, and let $\cB^{(0)} = \cB$. Then battlefield $l$ is assigned budget $\min(1,\cB^{(l-1)} \frac{v_l}{\sum_{j=l}^n v_j})$. So roughly speaking we assign to each battlefield a fraction of the budget equal to the fraction of the total value that the battlefield represents.
\end{remark}

\begin{remark}[Qualitative change in strategy as $k$ increases] \label{rem:increaseingkstrategychangebool}
We also qualitatively observe that as the number of players increases, the players are more likely to bid on battlefields of low value.

As an example, consider two battlefields with values given by $0 < v_2 < v_1 = 1$ and $k$ players with budget given by $\cB = 1$. Then (i) if $k\leq 1/v_2 + 1$,  we will have $p_1 = 1$ and $p_2 = 0$, meaning that if there are not enough players then no one will compete in the battlefield with small value. On the other hand, (ii) if $k > 1/v_2 + 1$, then we prove in Appendix~\ref{sec:boolproofs} that $p_2 > 0$, meaning that if there are enough players then they will compete in the low-value battlefield with some non-zero probability.
\end{remark}

As in the continuous-valued case, the General Lotto solutions in Lemma~\ref{lem:generallottoboolvalued} yield sufficient conditions for the players to be in Nash equilibrium:
\begin{lemma}\label{lem:suffcondsboolvalued}
Let $k \geq 3$. Suppose that in the symmetric Boolean-valued $k$-player Colonel Blotto game with battlefield valuations $v_1 \geq v_2 \geq \dots \geq v_n > 0$ and equal integer-valued budgets $\cB_i = \cB$, each player $i \in [k]$ independently bids a vector $A_{i,*} = (A_{i,1},\ldots,A_{i,n})$ such that for each $i \in [k]$:
\begin{enumerate}
    \item[(a)] $\sum_{j=1}^n A_{i,j} \leq \mathcal{B}$.
    \item[(b)] $A_{i,j} \sim \Ber(p_j)$, where $p_j$ is given in the statement of Lemma~\ref{lem:generallottoboolvalued}, using budget $\mathcal{B}$.
\end{enumerate}
Then the players are in equilibrium. Furthermore, this is the unique symmetric equilibrium.
\end{lemma}
The proof (in Appendix~\ref{sec:boolproofs}) is by linearity of expectation, as in the real-valued setting.

\subsection{Colonel Blotto equilibria}

We now show how to obtain an efficient Colonel Blotto strategy from the equilibrium General Lotto strategy. This consists of two tasks: (1) efficiently estimating the implicitly described $p_j$'s, and (2) efficiently computing a coupling of allocations that has the approximate $p_j$'s as its marginals. 
The first task, estimation, can be performed with a carefully tuned binary search. The second task presents an appealing puzzle: given $n$ Bernoulli random variables with biases $p_i,\ldots,p_n$ satisfying $\sum_{i=1}^n p_i = \cB \in \ZZ_{\geq 0}$, how can they be coupled into a joint distribution such that draws $x_1,\ldots,x_n \in \{0,1\}$ from the distribution satisfy $\sum_{i=1}^n x_i = \cB$ almost surely? Algorithm~\ref{alg:BooleanvaluedcolonelBlotto} is a very simple procedure for solving this puzzle.

\begin{algorithm}[H]
\SetAlgoLined
\DontPrintSemicolon
\SetKwInOut{Input}{Input}
\SetKwInOut{Output}{Output}
\SetKwFunction{ConstructM}{ConstructM}
\KwIn{A symmetric Boolean Blotto game $(k,n,\cB,\vec{v})$ with battlefield valuations $v_1 \geq v_2 \geq \dots \geq v_n > 0$.}
\BlankLine
\KwOut{A sample $(A_{1},\ldots,A_{n}) \in \RR^n$ from the equilibrium mixed strategy for a single player in the Boolean-valued Blotto game $(k,n,\cB, \vec{v})$.}
\BlankLine

\SetKwProg{Fn}{Function}{ is}{end}

For each $j \in [n]$, let $p_j$ be as defined in  Lemma~\ref{lem:generallottoboolvalued} or Theorem~\ref{thm:BooleanBlotto}.

For each $j \in [n]$, let $\alpha_j = \sum_{j'=1}^{j-1}p_{j'}$

Draw $\beta \sim \text{Unif}[0,1]$

For each $j \in [n]$, let $A_j = \one[ \exists m \in \mathbb{Z} \mid \beta + m \in [\alpha_j, \alpha_j + p_j)]$

\Return $(A_1,\ldots,A_n)$

 \caption{\texttt{NashEquilThm\ref{thm:BooleanBlotto}}: Boolean-valued Blotto equilibrium for Theorem~\ref{thm:BooleanBlotto}}
 \label{alg:BooleanvaluedcolonelBlotto}
 
\end{algorithm}

\begin{theorem}\label{thm:BooleanBlotto}
Suppose that in the Boolean-valued Colonel Blotto game with equal budgets $\cB_i = \cB$ and $k>2$ players, each player $i \in [k]$ independently runs Algorithm~\ref{alg:BooleanvaluedcolonelBlotto}. Then all of the players will be in Nash equilibrium.

Moreover, given parameter $\epsilon > 0$, Algorithm~\ref{alg:BooleanvaluedcolonelBlotto} runs in time polynomial in the problem size and $\log(1/\epsilon)$, and produces an $\epsilon$-approximate Nash equilibrium.
\end{theorem}

\begin{proof}

We verify that the sufficient conditions for a Nash equilibrium from Lemma~\ref{lem:suffcondsboolvalued} are met. We can assume without loss of generality that $\cB \leq n$, because otherwise all players compete in all battlefields, which is a Nash equilibrium. So in this case $\sum_{j=1}^n p_j = \cB \in \ZZ_{\geq 0}$.

An equivalent way of applying the sampling procedure is to set $$A_j = \one[\{\beta+m\}_{m \in \mathbb{Z}} \cap [\alpha_j, \alpha_j + p_j) \neq \emptyset].$$ Note that the intervals $[\alpha_j, \alpha_j+p_j)$ constitute a partition of the interval $[0,\cB)$, and that $\{\beta+m\}_{m \in \mathbb{Z}}$ intersects this long interval $\cB$ times, and finally that $\{\beta+m\}_{m \in \mathbb{Z}}$ intersects each interval $[\alpha_j, \alpha_j+p_j)$ at most once. It follows that, for any $\beta$, exactly $\cB$ of the $A_j$ bids are set to 1. This proves that the budget constraint holds almost surely.
And the probability that $A_j$ is set to 1 is $p_j$ because the interval $[\alpha_j, \alpha_i + p_j)$ has length $p_j$. So all the sufficient conditions of Lemma~\ref{lem:suffcondsboolvalued} are met.

\paragraph{Efficient approximation of equilibrium} We have constructed an exact Nash equilibrium. However, our algorithm is not yet efficient, because we have not yet described how to compute the probabilities $p_j$. These are defined implicitly in the statement of Lemma~\ref{lem:generallottoboolvalued}, but there appears to be no closed form. Nevertheless, if we could approximately compute the $p_j$ probabilities, then we could approximate the Nash equilibrium. Indeed, the utility for a player can range from $0$ to $V$ and there are $k$ players, so in order to compute an $\epsilon$-Nash equilibrium it suffices to approximate the equilibrium strategy for each player up to $(\epsilon/kV)$ error in statistical total variation. Since there are $n$ probabilities $p_j$, this can be achieved by approximating each $p_j$ up to additive error $(\epsilon/kVn)$. We want this estimation error even after scaling the approximate $p_j$'s so their sum is $\cB$; for this it suffices to achieve additive error $(\epsilon/kVn^2)$. We explain how to do with this with a carefully tuned binary search in a total number of $\poly(n,\log k, \log (V/\eps), \log(V/v_n))$ operations in Appendix~\ref{sec:boolproofs}.

\end{proof}

\section{Remarks \& Open Problems}

In this paper, we extended the definition of the Colonel Blotto problem to the multiplayer setting, and also introduced the study of the Boolean version of the problem.
We solved for the unique symmetric Lotto equilibria and coupled the marginals to construct Blotto equilibria in the symmetric case of these games under various parameter regimes of number of players, number of battlefields, and values of battlefields.
In all cases, we characterized the symmetric equilibria of the  General Lotto version of the game  and coupled the resulting bid distributions into a constrained joint distribution to solve the Blotto version.
A highlight of our paper is the efficient sampling algorithm for the symmetric three player case of continuous Blotto---Algorithm~\ref{alg:kequalsthree}---which is built from the geometric intuition of rotating a 2-sphere about the origin in hyperspace.
Interestingly, this result proves the existence of a coupling satisfying the sufficient constraints of Lemma~\ref{lem:colonelBlottosuffconds} by directly giving an algorithm to sample such a coupled distribution.
It is an open question whether the existence of the coupling can be proved in a more direct way.
This leads to our most general open question of characterizing when marginal distributions $\cD_1,\ldots,\cD_n$ over $\RR$ can be coupled into a joint distribution $\cD$ over $\RR^n$ such that a certain budget constraint holds almost surely in $\cD$. 
The decision problem is weakly $\NP$-hard even in the case of finitely-supported discrete distributions (by a simple reduction from \textsc{Subset-Sum}). It is an alluring problem to obtain a deeper understanding of the cases in which a budget-constrained coupling exists and can be constructed efficiently.

In Section~\ref{sec:multiplayerBlotto}, we gave an algorithm (Algorithm~\ref{alg:nequalskcolonelBlottonashequilibrium}) for efficiently sampling equilibrium strategies in the Blotto game for arbitrarily large numbers of players, as long as battlefields satisfied a value-partitioning constraint.
A special case captured by Corollary~\ref{cor:equalbudgets} is when all battlefields have equal value and the number of battlefields is a multiple of the number of players.
An important case left open therefore is solving for equilibria when the number of battlefields is arbitrary and there are four or more players.
A construction handling this case would complete the picture for symmetric homogeneous multiplayer Colonel Blotto. 

In Section~\ref{sec:BooleanBlotto}, we solved the multiplayer Boolean Blotto problem, where each player could play either a 0 or a 1 at each battlefield.
Of course, the generalization of this problem which allows players to make integer (not just Boolean) bids---discrete multiplayer Blotto---is a natural open problem.

\label{sec:remarks}

\bibliographystyle{plain}
\bibliography{bibliography}

\appendix
\section{Informal derivation of General Lotto solution}
Let us informally describe how we arrive at the General Lotto equilibria in Lemma \ref{lem:generallottocontinuousvalued}, assuming for simplicity that we are in the homogeneous setting considered by Myerson \cite{myerson}, so the battlefields have value 1. We are looking for an equilibrium that exploits the symmetry of the game across players and across battlefields. It is natural to guess that this can be achieved by all $k$ players playing the same single-variable distribution of bids on each of the $n$ battlefields. Denote the cumulative distribution function (CDF) of this distribution by $F$.

In order to derive $F$, we guess that $F$ has no atoms and is supported in a finite interval $[0,\theta]$. Then we consider what happens once players $1,\ldots,k-1$ have fixed their General Lotto strategies to playing $F$ on all $n$ battlefields. Suppose that player $k$ deviates and plays distributions $G_1,\ldots,G_n$ on the $n$ battlefields. Since $F$ has no atoms, a tie between the players is a measure-zero event, so the utility derived by player $k$ on battlefield $j$ is
$\PP[A_{k,j} > \max_{i \in [k-1]} A_{i,j}]$,  where $A_{1,j},\ldots,A_{k-1,j} \sim F$ and $A_{k,j} \sim G_j$ are independent.
Hence player $k$'s payoff on battlefield $j$ depends only on their bid relative to the maximum bid value $M_j = \max_{i \in [k-1]} A_{i,j}$ of all the other players.

Now, if $M_j$ is not uniform over $[0,\theta]$ for some $\theta$, then player $k$ can strictly gain over the other players by playing a slight perturbation $\tilde{F}$ of the distribution $F$, where $\epsilon$ probability mass is moved from values of $x$ where $\PP[M < x]/x$ is lower to values of $x$ where $\PP[M < x]/x$ is higher.
Therefore, if the players are in equilibrium, $\PP[M_j < x] = (F(x))^{k-1} = \min(1,\frac{x}{\theta})$, which implies that for all $i \in [k-1]$ we have
$$F(x) = \min\left(1,(x/\theta)^{\frac{1}{k-1}}\right).$$
One can solve for the scaling parameter $\theta$ by requiring that the budget constraint be tightly enforced: $\sum_{j=1}^n \E[A_{i,j}] = 1$ for any $i \in [k-1]$. We note that $F$ is a scaling of the $\mathrm{Beta}(1/(k-1),1)$ distribution.

\section{\texttt{RotatePair} correctness}
\begin{restatable}{claim}{rotatepairclaim}\label{claim:rotatepaircorrectness} \texttt{RotatePair} is correct and runs in $O(m)$ time.
\end{restatable}
\begin{proof}
If $\|u_1\|^2 = \|u_2\|^2$, then we must also have $\|u_1\|^2 = t_1 = t_2 = \|u_2\|^2$, so returning $(w_1, w_2) \leftarrow (u_1, u_2)$ is correct. Otherwise, items (i)-(iv) still hold:

(i) $\mathrm{supp}(w_1) \cup \mathrm{supp}(w_2) \subseteq \mathrm{supp}(u_1) \cup \mathrm{supp}(u_2)$ since $w_1,w_2$ are a linear combination of $u_1,u_2$.

(ii) $W = (\begin{array}{cc}w_1 & w_2\end{array}) \in \RR^{m \times 2}$ and $U = (\begin{array}{cc}u_1 & u_2\end{array}) \in \RR^{m \times 2}$ are related by $W = U \left(\begin{array}{cc} a & b \\ -b & a\end{array}\right)$, so $$WW^T = U \left(\begin{array}{cc} a & b \\ -b & a\end{array}\right) \left(\begin{array}{cc} a & -b \\ b & a\end{array}\right) U^T = U \left(\begin{array}{cc} a^2 + b^2 & 0 \\ 0 & a^2 + b^2\end{array}\right) U^T = UU^T,$$ since $a^2 + b^2 = 1$.

(iii and iv) There are two cases to consider. Note that since $u_1 \cdot u_2 = 0$, we have $\|w_1\|^2 = a^2 \|u_1\|^2 + b^2 \|u_2\|^2$ and $\|w_2\|^2 = b^2 \|u_1\|^2 + a^2 \|u_2\|^2$: \begin{itemize}
 \item If $\|u_1\|^2 - t_1 \geq t_2 - \|u_2\|^2$, then $$\|w_1\|^2 = \frac{(\|u_1\|^2 - t_2)\|u_1\|^2}{\|u_1\|^2 - \|u_2\|^2} + \frac{(t_2 - \|u_2\|^2)\|u_2\|^2}{\|u_1\|^2 - \|u_2\|^2}
 = \|u_1\|^2 + \|u_2\|^2 - t_2 \geq t_1$$ $$\|w_2\|^2 = \frac{(t_2 - \|u_2\|^2)\|u_1\|^2}{\|u_1\|^2 - \|u_2\|^2} + \frac{(\|u_1\|^2 - t_2)\|u_2\|^2}{\|u_1\|^2 - \|u_2\|^2} = t_2.$$
 \item If $\|u_1\|^2 - t_1 < t_2 - \|u_2\|^2$, then $$\|w_1\|^2 = \frac{(t_1 - \|u_2\|^2)\|u_1\|^2}{\|u_1\|^2 - \|u_2\|^2} + \frac{(\|u_1\|^2 - t_1)\|u_2\|^2}{\|u_1\|^2 - \|u_2\|^2}
 = t_1$$ $$\|w_2\|^2 = \frac{(\|u_1\|^2 - t_1)\|u_1\|^2}{\|u_1\|^2 - \|u_2\|^2} + \frac{(t_1 - \|u_2\|^2)\|u_2\|^2}{\|u_1\|^2 - \|u_2\|^2} = \|u_1\|^2 + \|u_2\|^2 - t_1 \geq t_2.$$
\end{itemize}
And in both cases conditions (iii) and (iv) hold.

The running time is $O(m)$, because we just compute the norms of two vectors of size $m$ and output a linear combination of the vectors.
\end{proof}

\section{Boolean Blotto Lemma Proofs}\label{sec:boolproofs}

For ease of presentation, we define:
$$\mu(p) = \begin{cases}
k-1, & p = 0 \\
\frac{1-(1-p)^{k-1}}{p}, & 0 < p \leq 1 \end{cases},$$
Thus, $m_v(p) = \frac{v}{k} \mu(p)$. The domain of $\mu^{-1}$ is extended by letting $\mu^{-1}(x) = 1$ for $x < 1$ and $\mu^{-1}(x) = 0$ for $x > k-1$. We also define the monotonically non-increasing function $$B(x) = \sum_{j=1}^n \mu^{-1}(k x/v_j),$$ and note that $x^*$ in Lemma~\ref{lem:generallottoboolvalued} is given by $\inf \{x \in \RR : B(x) \leq \cB\}$.

\begin{proof}[Proof of Claim~\ref{claim:continuous}]
By l'H\^{o}pital's rule $$\lim_{p \to 0^{+}} \mu(p) = \lim_{p \to 0^{+}} \frac{(k-1) (1-p)^{k-2}}{1} = k-1 = \mu(0),$$ proving continuity. And for any $p \in (0,1)$, $$\frac{\partial \mu}{\partial p} = \frac{p(k-1)(1-p)^{k-2} - 1 - (1-p)^{k-1}}{p^2} = \frac{(1-p)^{k-2}(1+p(k-2)) - 1}{p^2} < 0,$$ because for $t = k-2 > 0$ we have $(1-p)^{-t} \geq (1+pt)$. This holds because $(1-p)^{-t} |_{p=0} = 1 = (1+pt) |_{p = 0}$ and $\frac{\partial }{\partial p} (1-p)^{-t} = t(1-p)^{-t-1} \geq t = \frac{\partial }{\partial p} (1+pt)$ for $p \in [0,1]$.

The bijectivity follows from continuity and monotonicity.
\end{proof}

\begin{proof}[Proof of Remark~\ref{rem:increaseingkstrategychangebool}]

Part (i) follows because $B(1/k) = \sum_{j=1}^2 \mu^{-1}(1/v_j) \geq \mu^{-1}(1) = 1$, so $x^* \geq 1/k$, so $p_2 = \mu^{-1}(x^*/v_2) \leq \mu^{-1}(k-1) = 0$. To prove part (ii) consider $x = (k-1)v_2/k$, which satisfies $x > 1/k$ by the condition on the number of players. It follows that $B(x) = \mu^{-1}(kx) + \mu^{-1}(kx/v_2) = \mu^{-1}(kx) < 1$. By continuity of $B(x)$, there is $\epsilon > 0$ such that $B(x-\epsilon) \leq 1$, and therefore $x^* < x$. Hence $p_2 = \mu^{-1}(kx^*/v_2) > \mu^{-1}(k(k-1)v_2/(v_2 k)) = \mu^{-1}(k-1) = 0$.

\end{proof}

\begin{proof}[Proof of Lemma~\ref{lem:suffcondsboolvalued}]
The budget constraints are met by (a). If any player deviates from the strategy, then, by the analysis of Lemma~\ref{lem:generallottoboolvalued}, the player's expected payoff cannot improve. This is because by linearity of expectation the expected payoff for Colonel Blotto only depends on the marginal distributions of the bids for the battlefields.
\end{proof}

\subsection{Approximation procedure for Algorithm~\ref{alg:BooleanvaluedcolonelBlotto}}

\begin{enumerate}[wide, labelindent=0pt, itemsep=1em]
\item First, given any $x \in \RR$ we show how to compute an additive $\epsilon'$ approximation $\tilde{p}$ to $p = \mu^{-1}(x)$ in $\poly(\log k, \log(1/\epsilon'))$ operations. If $x \leq 1$ or $x \geq k-1$, then $\tilde{p} = 1$ or $\tilde{p} = k-1$ are respectively correct. Otherwise, for the case $1 < x < k-1$, recall from Claim~\ref{claim:continuous} that $\mu$ maps $[0,1]$ bijectively to $[1, k-1]$, and is continuous and monotonically decreasing. Therefore we can binary search to find $\tilde{p}$ such that $|\tilde{p} - p| < \epsilon'$. This binary search requires only $O(\log(1/\epsilon'))$ evaluations of $\mu$, and each evaluation of $\mu$ up to $\tilde{\epsilon}'$ precision costs only $\poly(\log k, \log(1/\tilde{\epsilon}'))$ operations. We can set the precision parameter to $\tilde{\epsilon}' = \epsilon'/2$, because for any $p',p'' \in [0,1]$, we have $|\mu(p') - \mu(p'')| \geq |p'-p''|$, since $\frac{d}{dp} \mu(p') \leq -1$ for all $p' \in [0,1]$. Therefore the total cost of the binary search is $\poly(\log k,\log(1/\epsilon'))$.

\item Second, we show how to compute $\tilde{x}$ such that $|\tilde{x} - x^*| < \epsilon''$, in $\poly(n,\log k,\log(V/\epsilon''))$ operations. Recall the definition $x^* = \inf \{x \in \RR : B(x) \leq \cB\}$, where $B(x) = \sum_{j=1}^n \mu^{-1}(kx/v_j)$. We will use the fact that $B(x)$ is monotonically non-increasing and continuous. By the proof of Lemma~\ref{lem:suffcondsboolvalued}, in the nontrivial case $\cB < n$ it holds that $x^* \in [0,(k-1)V]$. Hence we can binary search to find $\tilde{x}$ such that $|\tilde{x} - x^*| < \epsilon''$. The binary search requires $O(\log(kV/\epsilon''))$ evaluations of $B(x)$. Using part 1, each evaluation of $B$ up to precision $\tilde{\epsilon}''$ costs $\poly(n,\log k,\log(n/\tilde{\epsilon}''))$ operations, by separately evaluating each term up to precision $\tilde{\epsilon}''/n$. We now investigate the necessary precision $\tilde{\epsilon}''$. At any point in the binary search when we query point $\hat{x}$ one of two cases arises:
\begin{itemize}
\item Case A: For each $x'$ between $\hat{x}$ and $x^*$, there is a $j(x') \in [n]$ such that $\mu^{-1}(kx'/v_{j(x')}) \in (0,1)$. In this case, for all $x'$ between $\hat{x}$ and $x^*$, 
\begin{align*}
    \frac{dB(x)}{dx}|_{x = x'}    = \frac{d}{dx} \sum_{l=1}^n \mu^{-1}(kx/v_l)|_{x = x'} \leq \frac{d}{dx} \mu^{-1}(kx/v_{j(x')})|_{x = x'} \leq -\frac{2k}{(k^{2}-3k+2)v_{j(x')}} \leq -\frac{1}{Vk}
\end{align*}
So $|B(\hat{x}) - \cB| = |B(\hat{x}) - B(x^*)| \geq |\hat{x}-x^*|/(Vk^2),$ and so if $|\hat{x} - x^*| > \epsilon''/2$ it suffices to compute $B(\hat{x})$ up to accuracy $\tilde{\epsilon}'' = \epsilon''/(2Vk^2)$ in order to determine whether $\hat{x} \leq x^*$ or $\hat{x} > x^*$.
\item Case B: Otherwise there is $x'$ between $\hat{x}$ and $x^*$ such that $\mu^{-1}(kx'/v_j) \in \{0,1\}$ for all $j$. In this case, if $\hat{x} \leq x^*$ then our approximation $\tilde{B}(\hat{x})$ to $B(\hat{x})$ satisfies $\tilde{B}(\hat{x}) \geq |\{j : \mu^{-1}(k\hat{x}/v_j) = 1\}| \geq |\{j : \mu^{-1}(kx'/v_j) = 1\}| = B(x') \geq B(x^*)$. And by a similar argument $\tilde{B}(\hat{x}) > B(x^*)$  if $B(\hat{x}) > B(x^*)$; and $\tilde{B}(\hat{x}) \leq B(x^*)$ if $B(\hat{x}) \leq B(x^*)$; and $\tilde{B}(\hat{x}) < B(x^*)$ if $B(\hat{x}) < B(x^*)$.
\end{itemize}
Therefore we can set the precision parameter to $\tilde{\epsilon}'' = \epsilon''/(2Vk^2)$. So the total cost of the binary search is $\poly(n,\log k,\log(V/\epsilon''))$.

\item Third, suppose we have $\tilde{x}$ such that $|\tilde{x} - x^*| < \epsilon''$. Then for each $j \in [n]$ we define $\tilde{p}_j = \mu^{-1}(k\tilde{x}/v_j)$. By a simple calculation, $\mu^{-1}(x)$ is 1-Lipschitz over $\RR$, so we are guaranteed that $|\tilde{p}_j - p_j| \leq k\epsilon''/v_j \leq k\epsilon''/v_n$. Letting $\epsilon'' = (\epsilon v_n/ Vk^2n^2)/2$ and computing $\tilde{x}$ with the procedure from step 1, and approximating $\tilde{p}_j$ up to $\epsilon' = (\epsilon / V k^2n^2)/2$ error with the procedure from step 2, we obtain an overall $(\epsilon/Vkn^2)$ approximation to $p_j$. The total running time is $\poly(n, \log k,\log(V/\epsilon),\log(V/v_n))$, which is polynomial in the input size of the problem.

\item Finally, given approximations $\tilde{p}_j$ to the true $p_j$ probabilities, the sampling procedure takes time and space linear in $n$ and the number of bits of precision in the $\tilde{p}_j$ probabilities. This is polynomial in the problem size and $\log(1/\epsilon)$.
\end{enumerate}

\end{document}